\documentclass[11pt]{article}
\pdfoutput=1
\usepackage[margin=1.091in]{geometry}
\usepackage[utf8]{inputenc}
\usepackage[T1]{fontenc}
\usepackage[colorlinks]{hyperref}
\usepackage{mathtools, amsthm, amsfonts, amssymb}
\usepackage{accents} %
\usepackage{microtype}

\hypersetup{
  linkcolor=[rgb]{0.3,0.3,0.6},
  citecolor=[rgb]{0.2, 0.6, 0.2},
  urlcolor=[rgb]{0.6, 0.2, 0.2}
}

\usepackage{thmtools, thm-restate} 
\declaretheorem[name=Theorem, parent=section]{theorem}
\declaretheorem[name=Corollary, sibling=theorem]{corollary}
\declaretheorem[name=Proposition, sibling=theorem]{proposition}
\declaretheorem[name=Lemma, sibling=theorem]{lemma}

\theoremstyle{definition}

\theoremstyle{remark}

\DeclareMathAccent{\wtilde}{\mathord}{largesymbols}{"65}

\newcommand{\trk}{{\rm R}}
\newcommand{\subrank}{{\rm Q}}
\newcommand{\bsubrank}{{\underline{\rm Q}}}
\newcommand{\gensubrank}{{\rm Q}}
\newcommand{\spa}{{\rm span}}
\newcommand{\GL}{\operatorname{GL}}
\newcommand{\Mat}{\operatorname{Mat}}
\newcommand{\im}{{\rm im}}
\DeclareMathOperator{\asympsubrank}{\underaccent{\wtilde}{Q}}

\title{Generic Subrank}
\author{Harm Derksen, Visu Makam, and Jeroen Zuiddam}
\date{}

\begin{document}
\newgeometry{margin=1.4in,top=1.6in,bottom=1in}

\begin{center}
	{\LARGE Subrank and Optimal Reduction\\[0.2cm] of Scalar Multiplications to Generic Tensors}
\\[1cm] \large

\setlength\tabcolsep{0em}
\newcommand{\myPad}{\hspace{2em}}
\centerline{%
\begin{tabular}{c@{\myPad}c@{\myPad}c}
	Harm Derksen & Visu Makam & Jeroen Zuiddam\\[0.2em]
    \textsf{ha.derksen@northeastern.edu} & \textsf{visu@umich.edu} & \textsf{j.zuiddam@uva.nl}\\[0.2em]
     Northeastern University &  Radix Trading Europe B.V.  & University of Amsterdam
\end{tabular}%
}

\vspace{9mm}

\large
{\today}

\vspace{9mm}
\bf Abstract
\end{center}

\normalsize
\noindent
Since the seminal works of Strassen and Valiant it has been a central theme in algebraic complexity theory to understand the \emph{relative} complexity of algebraic problems, that is, to understand which algebraic problems (be it bilinear maps like matrix multiplication in Strassen's work, or the determinant and permanent polynomials in Valiant's) can be reduced to each other (under the appropriate notion of reduction). 

In this paper we work in the setting of bilinear maps and with the usual notion of reduction that allows applying linear maps to the inputs and output of a bilinear map in order to compute another bilinear map.
As our main result we determine precisely how many independent scalar multiplications can be reduced to a given bilinear map (this number is called the \emph{subrank}, and extends the concept of matrix diagonalization to tensors), for essentially all (i.e.~generic) bilinear maps. Namely, we prove for a generic bilinear map $T : V \times V \to V$ where $\dim(V) = n$ that $\theta(\sqrt{n})$ independent scalar multiplications can be reduced to $T$. Our result significantly improves on the previous upper bound from the work of Strassen (1991) and Bürgisser (1990) which was~$n^{2/3 + o(1)}$. Our result is very precise and tight up to an additive constant.
Our full result is much more general and applies not only to bilinear maps and 3-tensors but also to $k$-tensors, for which we find that the generic subrank is~$\theta(n^{1/(k-1)})$. Moreover, as an application we prove that the subrank is not additive under the direct sum.

The subrank plays a central role in several areas of complexity theory (matrix multiplication algorithms, barrier results) and combinatorics (e.g., the cap set problem and sunflower problem).
As a consequence of our result we obtain several large separations between the subrank and tensor methods that have received much interest recently, notably the slice rank (Tao, 2016), analytic rank (Gowers--Wolf, 2011; Lovett, 2018; Bhrushundi--Harsha--Hatami--Kopparty--Kumar, 2020), geometric rank (Kopparty--Moshkovitz--Zuiddam, 2020), and G-stable rank (Derksen,~2020).

Our proofs of the lower bounds rely on a new technical result about an optimal decomposition of tensor space into structured subspaces, which we think may be of independent interest. 

\thispagestyle{empty}
\newpage
\restoregeometry

\tableofcontents
\newpage

\section{Introduction}

We solve a fundamental problem in algebraic complexity theory about a notion of complexity on bilinear maps (tensors) called the \emph{subrank}, which was introduced by Strassen in \cite{strassen1987relative} in the study of fast matrix multiplication algorithms, and which later found close connections to several hypergraph independence and Ramsey-type problems in combinatorics and tensor methods in these areas (e.g.,~analytic rank \cite{lovett2019analytic,bhrushundi2018multilinear} and related notions). %
Our results improve significantly on previous bounds from the work of 
Bürgisser~\cite{burg} and Strassen \cite{strassen1991degeneration}.

In high-level terms, the subrank of a bilinear map is the largest number of independent scalar multiplications that can be reduced to (i.e.~``embedded~in'') the bilinear map, under the natural algebraic complexity notion of reduction (which we elaborate on in a moment). This definition extends the usual rank of matrices and moreover naturally extends further to multilinear maps ($k$-tensors).
In this paper, we:
\begin{itemize}
    \item determine the subrank for almost all bilinear maps (i.e.~generic bilinear maps),
    \item prove precise bounds that are accurate up to a small additive constant,
    \item extend the above results from bilinear maps ($3$-tensors) also to multilinear maps (\mbox{$k$-tensors}),
    \item prove as a technical result (used in the proofs of the above) an optimal decomposition theorem for tensor subspace, decomposing tensor space into very structured subspaces,
    \item prove, as an application of our upper bound on generic subrank, that the subrank is not additive under the direct sum.
\end{itemize} 

Let us biefly state our asymptotic results here. (We will return to these in more detail in \autoref{subsec:results}.) We prove for any vector space\footnote{We will require the base field of the vector space to have the mild property of being algebraically closed (but it will be clear from our techniques that this assumption can be weakened considerably).} $V$ with $\dim(V) = n$ that almost all bilinear maps $T : V \times V \to V$ have subrank equal to $\theta(\sqrt{n})$. That is, $\theta(\sqrt{n})$ independent scalar multiplications can be reduced to $T$. For $k$-tensors we find the subrank to be $\theta(n^{1/(k-1)})$ on almost all $k$-tensors.

To prove our results we use methods from algebraic geometry as well as linear algebraic arguments. Our upper bound proof relies on an efficient parametrization of the set of tensors with subrank larger than a given number (as a collection of orbits) and a good dimension estimate of this set.
Our lower bound proof relies on arguments involving differentials. At the core we prove a technical result about a very structured decomposition of tensor space which we think may be of independent interest. The main goal here (in the simplest case) is to write tensor space as a sum of tensor subspaces as efficiently as possible (meaning with smallest as possible sum of dimensions) such that each subspace has the special form of a matrix subspace tensored with $n$-space. Our technical result is that we can do this optimally for any order of tensor space. We discuss the proof methods further in \autoref{subsec:tech}.
To get some intuition for the subrank it is instructive to ask how the subrank relates to the better known complexity notion \emph{tensor rank}.
Subrank and rank are indeed closely linked, and in a sense they are dual to each other. Indeed, whereas the subrank of a bilinear map~$T$ is the maximal number of independent scalar multiplications that can be reduced to~$T$, the rank of $T$ is (among several equivalent definitions) the minimal number of independent scalar multiplications that $T$ reduces to. In this way, the rank measures the computational ``cost'' of~$T$ in terms of scalar multiplications while the subrank measures the ``value'' of $T$ in terms of independent scalar multiplications.

\subsection{Subrank and generic subrank}

Let us now discuss more precisely several equivalent formulations and the meaning of the definition of subrank, and the sense in which we determine the subrank for almost all tensors (the generic subrank).

\paragraph{Subrank of bilinear and multilinear maps.}
We first stay in the realm of bilinear maps, and will define subrank via the notion of a reduction between bilinear maps, which is really a reduction in the sense of computational complexity. Given two bilinear maps $S : V_1 \times V_2 \to V_3$ and $T : W_1 \times W_2 \to W_3$ we say $S$ reduces to $T$ and write $S\leq T$ if there are linear maps $L_1 : V_1 \to W_1$, $L_2 : V_2 \to W_2$ and $L_3 : W_3 \to V_3$ such that $S(v_1, v_2) = L_3(T(L_1(v_1), L_2(v_2))$ for all $v_1 \in V_1, v_2 \in V_2$. (In the literature, this relation $\leq$ on tensors is often called the \emph{restriction preorder} and when $S \leq T$ one says that ``$S$ is a restriction of~$T$\,'' \cite{strassen1987relative,burgisser1997algebraic}.) In other words, if $S \leq T$, then any algorithm for $T$ also gives an algorithm for $S$ (with only small computational overhead, namely applying the linear maps $L_i$ to inputs and output). Next, an important basic bilinear map that we use to define subrank is (for any positive integer $r$) the bilinear map~$I_r$ that computes $r$ independent scalar multiplications with scalars from some\footnote{We will put some mild restrictions on the field in some parts of the paper.} field $K$: 
\[
I_r : K^r \times K^r \to K^r : (v_1, v_2) \mapsto ( (v_1)_1 (v_2)_1, \ldots, (v_1)_r (v_2)_r ).
\]
With these notions set up, the subrank $\subrank(T)$ of a bilinear map $T : V_1 \times V_2 \to V_3$ is defined as the largest number $r$ such that $I_r$ reduces to $T$, that is, $I_r \leq T$.
On the other hand, the tensor rank of~$T$ is the smallest number $r$ such that $T \leq I_r$. Thus, the subrank $\subrank(T)$ is indeed the largest number of independent scalar multiplications that can be reduced to $T$, while the rank is the smallest number of independent scalar multiplications that $T$ reduces~to.
The above definition of subrank extends directly to multilinear maps $V_1 \times \cdots \times V_{k-1} \to V_k$ by naturally extending the reduction and defining 
$I_r$ as the multilinear map that computes~$r$ independent $(k-1)$-wise products.

\paragraph{Subrank of tensors.} Alternatively we may phrase the definition of the subrank in the language of tensors, as bilinear maps (and multilinear maps) naturally correspond to these. Let $K^{n_1, n_2, n_3}$ be the space of 3-tensors (3-dimensional arrays) $T = (T_{i,j,k})_{i,j,k}$ with $i \in [n_1]$, $j \in [n_2]$, $k \in [n_3]$. For tensors $S \in K^{n_1, n_2, n_3}$ and $T \in K^{m_1, m_2, m_3}$ we write $S \leq T$ if there are matrices $A \in \Mat_{n_1, m_1}$, $B \in \Mat_{n_2, m_2}$, $C \in \Mat_{n_3, m_3}$ such that $S$ is obtained from $T$ by applying $A,B,C$ to the slices of~$T$, in the sense that
\[
S_{i,j,k} = \sum_{a,b,c} A_{i,a} B_{j,b} C_{k,c} T_{a,b,c}
\]
for all $i \in [n_1], j \in [n_2], k \in [n_3]$. Let $I_r \in K^{r,r,r}$ be the tensor for which the diagonal entries $(I_r)_{i,i,i}$ are 1 and the other entries are 0. In this language the subrank $\subrank(T)$ is again the largest number $r$ such that $I_r \leq T$. This definition of subrank also naturally extends to higher-order tensors $K^{n_1, \ldots, n_k}$.

\paragraph{Subrank on ``almost all'' tensors; generic subrank.}
It follows from a short argument (which we give later; \autoref{prop:gen-subrank-exists}) that for vector spaces $V_1,V_2,V_3$, there is a subset $U$ of all bilinear maps $V_1 \times V_2 \to V_3$ that is nonempty and Zariski-open (the complement is the zero-set of a finite collection of polynomials) and with constant subrank. This set $U$ thus contains almost all bilinear maps $V_1 \times V_2 \to V_3$. The subrank on $U$ is called the \emph{generic subrank}. We denote the generic subrank of bilinear maps $K^{n_1} \times K^{n_2} \to K^{n_3}$ (or equivalently of tensors in $K^{n_1, n_2, n_3}$) by $\gensubrank(n_1, n_2, n_3)$ (and by $\gensubrank(n_1, \ldots, n_k)$ generally for higher-order tensors).

\subsection{Our Results}\label{subsec:results}

As our main results we determine (1) the subrank of generic tensors of order three, %
and, more generally, (2) the subrank of generic tensors of any order $k \geq 3$. %
Moreover, as a core technical ingredient in our proof of (1) and (2), we prove (3) an optimal decomposition of tensor space into highly structured subspaces, which we think is of independent interest and which may have further applications in algebraic complexity theory. We will now describe each of these results in detail.

\paragraph{(1) The generic subrank of tensors of order three.}
We will begin by discussing our results for tensors of order three (or equivalently, for trilinear maps %
or bilinear maps on vector spaces of appropriate dimensions). %
Recall that $\gensubrank(n) = \gensubrank(n,n,n)$ denotes the generic subrank of tensors in~$K^{n,n,n}$. In other words, $\gensubrank(n)$ is the value that the subrank takes on ``essentially all'' tensors, or on ``randomly chosen'' tensors with probability 1. We prove that~$\gensubrank(n)$ grows as~$\sqrt{n}$.

\begin{theorem}\label{thm:main-3-tensor}
We have $\gensubrank(n) = \theta(\sqrt{n})$.
\end{theorem}

We thus solve the problem of determining the generic subrank of tensors of order three (and also, with more work, of order $k$ in general as we will discuss below).
In computational terms, \autoref{thm:main-3-tensor} states that for ``essentially all'' bilinear maps $T : K^n \times K^n \to K^n$ we can reduce the problem of multiplying $\sqrt{n}$ independent pairs of numbers  to $T$, and that this is optimal (under the usual notion of reduction from algebraic complexity theory in which linear maps are applied to the inputs and output of $T$). 

In particular, with \autoref{thm:main-3-tensor} we significantly improve on the previous best upper bound on~$\gensubrank(n)$ of Strassen and Bürgisser (obtained via the method of ``lower support functionals'') which was $\gensubrank(n) \leq n^{2/3 + o(1)}$.

As an application of \autoref{thm:main-3-tensor} we prove that the subrank is not additive under the direct sum. Namely, we prove that there are bilinear maps $S,T : K^n \times K^n \to K^n$ such that $\subrank(T) + \subrank(S) = \theta(\sqrt{n})$ while $\subrank(T \oplus S) \geq n$. In other words, it is sometimes possible to reduce many more independent scalar multiplications to a direct sum of bilinear maps than to the bilinear maps separately.\footnote{This result is the subrank analogue of the recent result of Shitov \cite{MR3974478} that disproved Strassen's additivity conjecture for tensor rank \cite{MR521168}. For this he constructed a complicated example of 3-tensors $S,T$ (over any infinite field) such that $\trk(S \oplus T) < \trk(S) + \trk(T)$. We discuss additivity results further in \autoref{subsec:related}
.}

As a direct consequence of \autoref{thm:main-3-tensor}, we separate the generic subrank~$\gensubrank(n)$ from the generic asymptotic subrank $\asympsubrank(n)$\footnote{The asymptotic subrank is defined as $\asympsubrank(T) = \lim_{n\to\infty} \subrank(T^{\otimes n})^{1/n}$ and $\asympsubrank(n)$ denotes the value of $\asympsubrank(T)$ for a generic tensor $T \in K^{n,n,n}$.}, for which Strassen proved that the generic value satisfies~$\asympsubrank(n) \geq n^{2/3}$ \cite[Prop.~3.6]{strassen1988asymptotic}. Moreover, our result gives a large separation between the subrank on the one hand, and the slice rank (and partition rank, for higher order tensors), geometric rank, and G-stable rank on the other hand (whose generic value is full).

Let us now discuss the bound of \autoref{thm:main-3-tensor} in more detail, and in particular get the precise constants right, after which we discuss the extension to higher-order tensors.
We (naturally) obtain \autoref{thm:main-3-tensor} in two parts, namely by first proving the following upper bound:

\begin{theorem} \label{thm:subrank-upper-3-tensor}
We have 
$\gensubrank(n) \leq \lfloor\sqrt{3n-2}\rfloor$.
\end{theorem}

And then proving the following essentially matching lower bound:

\begin{theorem}\label{thm:subrank-lower-3-tensor}
We have $\gensubrank(n) \geq 3(\lfloor\sqrt{n/3 + 1/4}-1/2\rfloor)$.
\end{theorem}

It is not difficult to see that our upper and lower bounds on the generic subrank are very precise. Namely, 
the difference between the upper bound~$\lfloor\sqrt{3n-2}\rfloor$ %
and the lower bound~$3(\lfloor\sqrt{n/3 + 1/4}-1/2\rfloor)$ %
is at most a small additive constant.\footnote{Indeed, a straightforward direct computation shows that  $\lfloor\sqrt{3n-2}\rfloor - 3(\lfloor\sqrt{n/3 + 1/4}-1/2\rfloor) \leq 5$.}

We conjecture that our upper bound in \autoref{thm:subrank-upper-3-tensor} is tight.
Through a more sophisticated analysis (which requires as a final component a computer verification that a determinant is nonzero) we prove that for all $1\leq n \leq 100$ our upper bound is in fact tight, that is: for all $1 \leq n \leq 100$ we have $\gensubrank(n) = \lfloor\sqrt{3n-2}\rfloor$. %

\paragraph{(2) The generic subrank of $k$-tensors.}
So far we have discussed only tensors of order three. With more elaborate methods we are able to completely extend our above results from tensors of order three to tensors of order $k$ for every $k \geq 3$. Denoting the subrank of a generic tensor in~$K^{n,\ldots, n}$ of order $k$ by $\gensubrank(n,\ldots,n)$, we find that $\gensubrank(n,\ldots,n)$ grows as~$n^{1/(k-1)}$.

\begin{theorem}\label{thm:main-k-tensor}
We have $\gensubrank(n,\ldots,n) = \theta(n^{1/(k-1)})$.
\end{theorem}

Again this result directly leads to a large separation between the subrank and the asymptotic subrank (for which the generic value, extending a construction of Strassen, satisfies $\asympsubrank(n) \geq n^{2/k}$), and a separation between the subrank and the slice rank, partition rank, geometric rank and G-stable rank (all of whose generic value is $n$).

Regarding the precise bounds, we prove \autoref{thm:main-k-tensor} in two parts again. 
In the first we extend the upper bound of \autoref{thm:subrank-upper-3-tensor}
to order $k$, in a fully general manner as an upper bound on the generic subrank $\gensubrank(n_1,\ldots,n_k)$ where the $n_i$ need not be equal, as follows.

\begin{theorem} \label{thm:subrank-upper-k-tensor}
We have
$
\gensubrank(n_1,\ldots,n_k) \leq \bigl(\sum_{i=1}^k n_i - (k-1)\bigr)^{\frac{1}{k-1}}.
$
\end{theorem}

Then finally we extend \autoref{thm:subrank-lower-3-tensor} to order $k$ tensors, which leads to the asymptotically matching lower bound (which for conciseness we will not write down here more explicitly).

\begin{theorem}\label{thm:subrank-lower-k-tensor}
We have $\gensubrank(n,\ldots,n) \geq \Omega(n^{1/(k-1)})$.
\end{theorem}

Our proof of this last theorem (and also \autoref{thm:subrank-lower-3-tensor}) makes crucial use of the  technical results on tensor space that we discuss in the next section.

\paragraph{(3) Technical result: structured subspace decomposition of tensor space.}

The proofs of our lower bounds (\autoref{thm:subrank-lower-3-tensor} and \autoref{thm:subrank-lower-k-tensor}) rely on a technical result about a very structured decomposition of tensor space which we think may be of independent interest, so we will describe it here. The main goal here (in the simplest case) is to write tensor space~$K^{n, n, n}$ as a sum of tensor subspaces as efficiently as possible (meaning with smallest as possible sum of dimensions) such that each subspace has the form of an $n \times n$ matrix subspace tensored with~$K^n$ (so that it becomes a subspace of $K^{n,n,n}$) with the tensoring being applied in any of the three possible directions. Our result is that we can do this optimally for any order of tensor space.

We will begin our discussion with the simplest version, which is for the tensor space~$K^{3,3,3}$, as it is the easiest to explain (and in fact also not hard to prove) and forms the basis for the proof of the (harder to prove) general version for the tensor space $K^{n,\ldots, n} = (K^n)^{\otimes n}$ of tensors of order~$n$ and dimension~$n$ in each direction (in our application to the generic subrank we use a blow-up argument so that we can deal with tensors of order $k$ and dimension $n$ in each direction). 

Recall that $K^{n_1, n_2, n_3}$ denotes the space of $n_1 \times n_2 \times n_3$ tensors with coefficients in the field~$K$, which we require to be large enough in this part. Let $\Mat_{n,m}$ denote the space of~$n \times m$ matrices with coefficients in $K$.
We use a special notation to denote a certain construction of a tensor subspaces given a matrix subspace. Namely, for any given $n_1 \times n_2$ matrix subspace $\mathcal{W} \subseteq \Mat_{n_1, n_2}$ we denote by $\mathcal{W}[3]$ the tensor subspace $\mathcal{W} \otimes K^{n_3} \subseteq K^{n_1, n_2, n_3}$. Analogously, for any matrix subspace $\mathcal{W} \subseteq \Mat_{n_2, n_3}$ we denote by $\mathcal{W}[1] \subseteq K^{n_1, n_2, n_3}$ the tensor subspace obtained by appropriately tensoring $\mathcal{W}$ with $K^{n_1}$, and for any matrix subspace $\mathcal{W} \subseteq \Mat_{n_1, n_3}$ we denote by $\mathcal{W}[2] \subseteq K^{n_1, n_2, n_3}$ the tensor subspace obtained by appropriately tensoring $W$ with~$K^{n_2}$.

For the tensor space $K^{3,3,3}$ we have the following optimal decomposition theorem.

\begin{theorem}\label{thm:space-decomp-333}
There exist subspaces $\mathcal{X}_i \subseteq \Mat_{3,3} = K^3 \otimes K^3$, each of dimension 3, such that 
\[%
K^{3,3,3} = \mathcal{X}_1[1] + \mathcal{X}_2[2] + \mathcal{X}_3[3].\footnote{We note that the existence of subspaces $\mathcal{X}_i$ with this property is equivalent to generic subspaces $\mathcal{X}_i$ having this property (i.e.~there being a non-empty Zariski-open set of triples $\mathcal{X}_i$ with this property).}
\]%
\end{theorem}
Comparing dimensions, we have $\dim K^{3,3,3} = 3\cdot3\cdot 3 = \dim \mathcal{X}_1[1] + \dim \mathcal{X}_2[2] + \dim \mathcal{X}_3[3]$ and so the decomposition in \autoref{thm:space-decomp-333} is optimal. In other words, it is a direct sum decomposition of~$K^{3,3,3}$. 

The requirement in \autoref{thm:space-decomp-333} that the $\mathcal{X}_i$ all have dimension 3 is crucial to make the theorem interesting, as without this requirement we could ``decompose'' $K^{3,3,3}$ simply as $K^{3,3,3} = \mathcal{X}_1[1]$ with $\mathcal{X}_1 = \Mat_{3,3}$. 
Interestingly, the analogous statement of \autoref{thm:space-decomp-333} for any matrix space~$K^{n,n}$ is false. That is, for every positive integer $n$ there are no subspaces $\mathcal{X}_i \subseteq K^n$, each of dimension~$n/2$, such that $K^{n,n} = \mathcal{X}_1[1] + \mathcal{X}_2[2]$.
This is saying that if we pick a row space $R$ and a column space~$C$ of dimension~$n/2$, then it is not possible to write every $n\times n$ matrix~$A$ as $A_1 + A_2$ where the row space of $A_1$ is in $R$ and the column space of $A_2$ is in $C$.
On the other hand, 3-tensors do not suffer from this malady: we can write any tensor  as a sum of thee tensors, whose ``row${}_i$''-space are asked to be in a generic space of the right dimension. (Here we can think of $W[i]$ as tensors whose ``row${}_i$''-space is $W$.)
Therefore, \autoref{thm:space-decomp-333} is demonstrating an interesting phenomenon that occurs in tensor space but not in matrix space.

\autoref{thm:space-decomp-333} (as opposed to the upcoming generalization to $(K^n)^{\otimes n}$) is not difficult to prove. Indeed, one may choose the matrix subspaces $\mathcal{X}_i$ randomly and then for such an explicit choice verify directly that they satisfy the claim (and this approach will work with high probability). In fact, there are even valid choices of the $\mathcal{X}_i$ such that each $\mathcal{X}_i$ is spanned by a matrix with coefficients in $\{0,1\}$.

Next we discuss the higher-dimensional version of \autoref{thm:space-decomp-333} in which $K^{3,3,3}$ is generalized to $n$-tensor space $(K^n)^{\otimes n}$. We naturally extend our previous notation so that for every tensor subspace $\mathcal{W} \subseteq (K^n)^{\otimes (n-1)}$ we define, for any $i \in \{1, \ldots, n\}$, the tensor subspace $\mathcal{W}[i] \subseteq (K^n)^{\otimes n}$ by tensoring $\mathcal{W}$ with $K^n$ in one of the $n$ possible ways. 

By a recursive construction with $K^{3,3,3}$ as a base case, we find the following optimal decomposition of $n$-tensor space $(K^n)^{\otimes n}$ for all $n\geq3$.

\begin{theorem}\label{thm:space-decomp-n}
For every integer $n\geq 3$ there exist subspaces $\mathcal{X}_i \subseteq (K^n)^{\otimes n-1}$ of dimension~$n^{n-2}$ such that
\[
(K^n)^{\otimes n} = \mathcal{X}_1[1] + \mathcal{X}_2[2] + \cdots + \mathcal{X}_n[n].
\]
\end{theorem}
Again, since $\dim K^{n,\ldots, n} = n^n = n \cdot n\cdot n^{n-2} = \sum_{i=1}^n \dim \mathcal{X}_i[i]$, the decomposition in \autoref{thm:space-decomp-n} is optimal in terms of dimension and hence a direct sum decomposition of~$(K^n)^{\otimes n}$.

\autoref{thm:space-decomp-n} is the theorem we use to prove the general generic subrank lower bound \autoref{thm:subrank-lower-k-tensor}. However, the methods we introduce in the process of proving \autoref{thm:space-decomp-n} allow us much more generally for other choices of positive integers $n_1, \ldots, n_k$ to construct optimal decompositions $K^{n_1, \ldots, n_k} = \mathcal{X}_1[1] + \mathcal{X}_2[2] + \cdots + \mathcal{X}_k[k]$ from known decompositions. This leads to a natural fundamental mathematical question (with potentially other applications) of what choices of $n_1, \ldots, n_k$ and $\dim \mathcal{X}_i$ allow such decompositions.

\subsection{Technical Overview}\label{subsec:tech}

We give a brief technical overview of the methods and ideas that we use in our proofs. 

\paragraph{Upper bounds on generic subrank.}

The high-level approach in our proof of the upper bound on the generic subrank in \autoref{thm:subrank-upper-3-tensor} (and similarly for the general case of \autoref{thm:subrank-upper-k-tensor}) is as follows. For any nonnegative integer we consider the set $\mathcal{C}_r$ of tensors in $K^{n,n,n}$ with subrank at least~$r$. We argue that the generic subrank is precisely the largest $r$ such that the dimension of~$\mathcal{C}_r$ equals the dimension $n^3$ of the full space $K^{n,n,n}$. We then prove the core ingredient, namely the dimension upper bound $\dim(\mathcal{C}_r) \leq n^3 - r(r^2 - 3n + 2)$. This information leads to the desired result, since if we let $t$ be the generic subrank $\gensubrank(n)$, then we must by the above have $n^3 = \dim(\mathcal{C}_t) \leq n^3 - t(t^2 - 3n + 2)$, from which we directly deduce that $t\leq \sqrt{3n - 2}$, that is, we obtain the bound $\gensubrank(n) \leq \sqrt{3n-2}$ of \autoref{thm:subrank-upper-3-tensor}.

To prove the aforementioned dimension upper bound on $\mathcal{C}_r$ that is the core ingredient in the above argument we employ the idea of providing a (non-injective) parametrization of $\mathcal{C}_r$, compute the dimension of the parameter space, and then subtracting the dimension ``over-count'' (the fiber dimension under the parametrization). 
For this we first define a set $X_r$ of tensors in $K^{n,n,n}$ of a special form, namely whose $[r]\times [r]\times [r]$ subtensor is zero except for the diagonal which is nonzero. Then the elements of $X_r$ clearly have subrank at least $r$ (and are thus in $\mathcal{C}_r$). %
The important point is that, by applying all possible basis transformations to the tensors in $X_r$ we obtain all of~$\mathcal{C}_r$.
Thus $X_r$ together with the group of all basis transformations provide the parametrization of~$\mathcal{C}_r$.
Technically, we describe this by saying that the map $\psi_r : \GL_n \times \GL_n \times \GL_n \times X_r \to K^{n,n,n}$ that maps $(A, B, C, T)$ to $(A \otimes B \otimes C)T$ has image precisely $\mathcal{C}_r$. Now the computation to upper bound $\dim(\mathcal{C}_r)$ consists of computing the dimension of the domain $\GL_n \times \GL_n \times \GL_n \times X_r$ and subtracting the dimension of a general fiber of $\psi_r$, which we carry out to arrive at the dimension upper bound stated earlier.

\paragraph{Lower bounds on generic subrank.}

The high-level approach in our proof of the lower bound on the generic subrank in \autoref{thm:subrank-lower-3-tensor} (and the general \autoref{thm:subrank-lower-k-tensor}) is as follows. We use notation defined in the upper bound proof  discussion and the results section (\autoref{subsec:results}). Our proof reduces the problem of lower bounding the generic subrank $\gensubrank(n_1, n_2, n_3)$ to a problem of constructing tensor space decompositions of a specific form (which we discuss further in the next section). Namely we prove that, for $r \leq n_1, n_2, n_3$ (with a technical condition), if $\mathcal{X}_i \subseteq \Mat_{r,r}$ are subspaces of dimension $n_i - r$ for $i=1,2,3$ such that 
\[
\mathcal{X}_1[1] + \mathcal{X}_2[2] + \mathcal{X}_3[3] = K^{r,r,r},
\footnote{where $\mathcal{X}_1[1] \subseteq K^{r,r,r}$ denotes $\mathcal{X}_1$ tensored with $K^r$ in the first tensor leg, $\mathcal{X}_2[2]$ denotes $\mathcal{X}_2$ tensored with $K^r$ in the second tensor leg, and $\mathcal{X}_3[3]$ denotes $\mathcal{X}_3$ tensored with $K^r$ in the third tensor leg}
\]
then $\gensubrank(n_1, n_2, n_3) \geq r$. (And we prove the analogous statement for higher-order tensors.) Finding such $\mathcal{X}_i$ we discuss in the next section. The proof of the lower bound given the $\mathcal{X}_i$ goes as follows.

Recall the map $\psi_r : \GL_n \times \GL_n \times \GL_n \times X_r \to K^{n,n,n}$ whose image we already claimed is the set $\mathcal{C}_r$ of tensors of subrank at least $r$. To reach our goal we want to find conditions that imply that the image of $\psi_r$ has full dimension $n^3$ and thus is Zariski-dense in $K^{n,n,n}$.
To do this we use the notion of the differential $d\psi_r$ of $\psi_r$ and a general method that says that the dimension of a map can be computed as the rank of the differential at a ``generic point''.

The differential $d\psi_r$ at the point $(g_1, g_2, g_3, T)$ is the map 
\[
(d\psi_r)_{(g_1, g_2, g_3, T)} : \Mat_{n, n} \times \Mat_{n,n} \times \Mat_{n,n} \times Y_r \to K^{n,n,n}
\]
where $Y_r$ is the tangent space of $X_r$, given by
\[
(A, B, C, S) \mapsto ((A \otimes g_2 \otimes g_3) + (g_1 \otimes B \otimes g_3) + (g_1 \otimes g_2 \otimes C))T + (g_1 \otimes g_2 \otimes g_3)S.
\]
Analyzing this map we compute its image which leads to the aforementioned lower bound statement on the generic subrank.
In fact we prove a stronger lower bound, which in characteristic~0 characterizes the generic subrank precisely, but with a harder to analyze condition, in \autoref{thm:technical-lower-bound-strong}. In particular, if the characteristic of $K$ is $0$, then \autoref{thm:technical-lower-bound-strong} says that $\gensubrank(n_1,n_2,n_3)$ is given by the smallest number $r$ such that 
$\mathcal{X}_1[1] + \mathcal{X}_2[2] + \mathcal{X}_3[3] + W_r = K^{r,r,r}$ for generic subspaces $\mathcal{X}_i \subseteq \Mat_{r,r}$ of dimension $n_i -r$, where $W_r \subseteq K^{r,r,r}$ is the subspace of tensors such that $T_{ijk} = 0$ if  $i,j,k$ are all different. (In other characteristics this number $r$ gives a lower bound.)

\paragraph{Constructions of tensor space decompositions.}

To prove %
\autoref{thm:space-decomp-n} we introduce general methods to construct the optimal tensor space decompositions as described in the theorem from existing ones. We then give some small constructions and combine these in multiple recursions to achieve the required object.

To set this up we take a very general approach in which we study direct sum decompositions 
\[
K^{n_1, n_2, \ldots, n_k} = \mathcal{X}_1[1] + \cdots + \mathcal{X}_k[k]
\]
where $\mathcal{X}_i \subseteq K^{n_1} \otimes \cdots \otimes K^{n_{i-1}} \otimes K^{n_{i+1}} \otimes \cdots \otimes K^{n_k}$ is a tensor subspace and $\mathcal{X}_i[i]$ denotes the subspace obtained by tensoring~$\mathcal{X}_i$ with $K^{n_i}$ as the $i$th tensor factor. Let $a_i$ be the dimension of~$\mathcal{X}_i$. We are interested in which values of $n_1, \ldots, n_k$ and $a_1, \ldots, a_k$ allow for a decomposition of the above form. Writing these numbers into a $2\times k$ matrix
\[
\begin{bmatrix} n_1 & n_2 & \cdots & n_k\\a_1 & a_2  & \cdots & a_k\end{bmatrix}
\]
we let $\mathcal{S}$ be the set of all such matrices for which a decomposition exists satisfying the parameters given in the matrix. Then \autoref{thm:space-decomp-n} corresponds to proving that the $2 \times n$ matrix
\begin{equation}\label{eq:base-case-n-intro}
\begin{bmatrix} n & n & \cdots & n\\n^{n-2} & n^{n-2}  & \cdots & n^{n-2}\end{bmatrix}
\end{equation}
is an element of $\mathcal{S}$. 

Next we observe that there are some simple constructions and properties of elements in $\mathcal{S}$, such as: if a matrix is in~$\mathcal{S}$ then if we permute its columns it is still in~$\mathcal{S}$ and the matrix~$\begin{bsmallmatrix} n\\1 \end{bsmallmatrix}$ is in~$\mathcal{S}$. With slightly more work we can give direct constructions for
\begin{equation}\label{eq:base-case-intro}
\begin{bmatrix} 2 & 2 & 2\\0 & 3 & 1\end{bmatrix}
\end{equation}
being an element of $\mathcal{S}$, for instance.

In order to construct more elements of $\mathcal{S}$ we prove a ``direct sum construction'' that given two elements in $\mathcal{S}$ combines them to get a new one. Namely, this result gives that, if
\[
\begin{bmatrix} n_1 & n_2 & \dots & n_{k-1} & n_k'\\
a_1' & a_2' & \cdots & a_{k-1}' & a_k\end{bmatrix} \in \mathcal{S}
\]
and
\[
\begin{bmatrix} n_1 & n_2 & \dots & n_{k-1} & n_k''\\
a_1'' & a_2'' & \cdots & a_{k-1}'' & a_k\end{bmatrix} \in \mathcal{S}
\]
then we have 
\[
\begin{bmatrix} 
n_1 & n_2 & \cdots & n_k\\
a_1 & a_2  & \cdots & a_k
\end{bmatrix}\in \mathcal{S}
\]
where $n_k=n_k'+n_k''$, and $a_i=a_i'+a_i''$ for $i=1,2,\dots,k-1$.

Finally, from some simple base cases including \eqref{eq:base-case-intro} we give a construction for 
\[ %
\begin{bmatrix} 3 & 3 & 3\\3 & 3 & 3\end{bmatrix}
\] %
being an element of $\mathcal{S}$
and via an elaborate argument with multiple recursions arrive at the matrix in \eqref{eq:base-case-n-intro} being an element of $\mathcal{S}$, which is the ingredient required in our proof of \autoref{thm:subrank-lower-k-tensor}. It is natural to ask what precisely are all elements of~$\mathcal{S}$, which we leave as an open problem.

\subsection{Related Work}\label{subsec:related}

The previous best bound on the generic subrank $\gensubrank(n)$ was the upper bound $n^{2/3 + o(1)}$ which follows from the work of Bürgisser \cite[Satz~2.8]{burg}\footnote{\cite[Satz~2.8]{burg} determines the generic value of a tensor parameter called the ``lower support functional'' (unteren  Trägerfunktional) which upper bounds the subrank as proven in \cite{strassen1991degeneration}.} as part of the broader research program on the theory of asymptotic spectra of tensors (motivated by the study of matrix multiplication algorithms). The proof of this bound relies on the method of ``lower support functionals'' introduced by Strassen in \cite{strassen1991degeneration} (see also the more recent surveys on this topic in \cite{christandl2017universalproc,phd,WZ}) and the properties of these that he proves there. This method recovers certain asymptotic information about tensors, which importantly is monotone under the restriction preorder and normalized on diagonal tensors so that it provides an upper bound on the subrank (in a manner that is very different from the approach that we take to prove our optimal upper bound). Bürgisser's analysis of this method on generic tensors consists of proving that the support of a generic tensor is large for any choice of basis and a combinatorial study of a certain type of covering of these supports, which leads to the aforementioned $n^{2/3 + o(1)}$ upper bound.

Recent research has brought about a rich collection of tensor methods that are in a similar ``regime'' as the subrank (for instance, they are all monotone under restriction), each with their own properties and applications in complexity theory and combinatorics. Notable are the slice rank~\cite{tao} and closely related partition rank \cite{naslund2017multi}, the analytic rank \cite{MR2773103,lovett2019analytic,bhrushundi2018multilinear}, the geometric rank \cite{DBLP:conf/coco/KoppartyMZ20, geng2021geometry} and G-stable rank \cite{derksen2020gstable}. Some important applications of these methods include new bounds on cap sets \cite{tao,MR3583358}, new bounds on the sunflower problem~\cite{MR3668469}, determining the border subrank of matrix multiplication \cite{DBLP:conf/coco/KoppartyMZ20}, and proving matrix multiplication barriers \cite{DBLP:journals/corr/BlasiakCCGU16,DBLP:conf/coco/Alman19,DBLP:conf/coco/ChristandlVZ19}. Many strong connections have been shown among these parameters. In particular, Derksen \cite{derksen2020gstable} showed that the G-stable rank is equal to the slice rank up to a constant factor, and Cohen--Moshkovitz \cite{DBLP:conf/stoc/CohenM21,cohen2021partition} showed (over large fields) that the partition rank, analytic rank and geometric rank are equal up to a constant factor, the culmination of a long line of work on this topic \cite{DBLP:journals/cdm/GreenT09,DBLP:conf/focs/KaufmanL08,bhowmick2015bias,janzer2018low,janzer2019polynomial,milicevic2019polynomial}.
All of the aforementioned tensor parameters are lower bounded by the subrank.\footnote{The analytic rank requires a natural normalization for this to be true, that is, one should normalize it by the analytic rank of a full-dimensional diagonal tensor.} Our results provide a large separation on almost all tensors between the subrank and all other aforementioned parameters, as the generic subrank satisfies $\gensubrank(n) = \theta(\sqrt{n})$ whereas the generic value of all other parameters is the maximal value~$n$.

The study of ``generic'' or ``typical'' complexity in algebraic complexity theory goes back to at least Strassen's paper ``Rank and optimal computation of generic tensors'' \cite{MR709378}, in which he determines the tensor rank of almost all tensors (i.e.~generic tensors). His result is that the rank $\trk(n_1,n_2,n_3)$ of almost all tensors in $K^{n_1,n_2,n_3}$ grows as $n_1 n_2 n_3/ (n_1 + n_2 + n_3 - 2)$ and a description is given of ``perfect shapes'' $(n_1, n_2, n_3)$ for which precisely equality $\trk(n_1,n_2,n_3) = n_1 n_2 n_3/ (n_1 + n_2 + n_3 - 2)$ holds. (It is still a fundamental open problem to find explicit tensors in~$K^{n,n,n}$ with rank close to $n^2$ \cite{blaser2014explicit} with important implications to formula lower bounds \cite{raz2013tensor}.) For the shape $(n,n,3)$ with $n$ odd this work provides an equation that determines whether a tensor has typical rank or not. These equations lead to several later generalizations \cite{MR3081636,koiran2020tensor} and subsequently precise barrier results for ``rank methods'' \cite{efremenko2017barriers,DBLP:conf/focs/GargMOW19}. Our result rather than generic rank determines the generic subrank. The rank of a tensor being the smallest number of scalar multiplications that the tensor reduces to, the subrank is a natural dual to the tensor rank. We note that combining the results on generic subrank and generic rank, we see that reducing scalar multiplications to a generic tensor and reducing the generic tensor back to scalar multiplications necessarily induces a great loss, as $\gensubrank(n,n,n) = \theta(\sqrt{n})$ is much smaller than $\trk(n,n,n) = \theta(n^2)$.

The study of additivity results is a central theme in complexity theory and mathematics. In algebraic complexity, it has been long known (and crucial in the design of matrix multiplication algorithms) that the border rank of tensors (the approximative version of tensor rank) is not additive under the direct sum \cite{schonhage1981partial} (see also \cite[Lemma~7.1]{blaser2013fast} and \cite[(15.12)]{burgisser1997algebraic}). Strassen conjectured the tensor rank to be additive under the direct sum, but this was disproved recently by Shitov \cite{MR3974478}. On the other hand, the analytic rank \cite{lovett2019analytic}, geometric rank \cite{DBLP:conf/coco/KoppartyMZ20}, G-stable rank \cite{derksen2020gstable} and slice rank \cite{gowers2021slice} were shown to be additive recently. The subrank, as we show, is however not additive. Our proof of this relies on writing a tensor as a sum of two generic tensors, which is reminiscent of methods that Razborov \cite{razborov1992submodular} uses to prove a linear upper bound on submodular complexity measures of boolean functions (see also \cite{DBLP:journals/eccc/RobereZ21}).

\subsection{Paper Organization}

In \autoref{sec:upper-bounds} we prove the upper bound theorems \autoref{thm:subrank-upper-3-tensor} and \autoref{thm:subrank-upper-k-tensor}. In \autoref{sec:lower-bounds} we prove the characterization of generic subrank that allows us to reduce the problem of determining its value to the tensor subspace decomposition problem. In \autoref{sec:constructions} we solve the tensor subspace decomposition problem thus completing the lower bound proof. 
In \autoref{sec:not-additive} we use our upper bounds to prove that the subrank is not additive under direct sum.
In \autoref{sec:open-problems} we discuss several natural related open problems.

\newpage

\section{Upper bounds on generic subrank}\label{sec:upper-bounds}
In this section, we prove our upper bounds on the generic subrank of tensors. We give a detailed proof in the case of $3$-tensors, and then generalize to tensors of all orders.

\subsection{Tensors of order three}

The techniques we use are familiar in invariant theory and representation theory. The main idea is to take advantage of the fact that subrank is invariant under a large group of symmetries, namely change of basis on each tensor leg. Further, this group of symmetries has excellent algebraic properties which can often be leveraged to remarkable effect.

First and foremost, we have to argue that generic subrank is a valid notion, which we do in the following proposition. The proof is a standard argument which we will discuss because it naturally uses some ingredients that we will use later on.
\begin{proposition}\label{prop:gen-subrank-exists}
For every $n$ there is a non-empty Zariski-open subset $U \subseteq K^{n,n,n}$ and integer~$r$ such that for all $T \in U$ we have $\subrank(T) = r$.
\end{proposition}

For any $n$, the number $r$ given by \autoref{prop:gen-subrank-exists} is unique since any two non-empty Zariski-open subsets $U_1, U_2 \subseteq K^{n,n,n}$ must intersect ($K^{n,n,n}$ is irreducible). This number $r$ we call the generic subrank $\subrank(n)$. Similarly we define the generic subrank $\subrank(n_1, \ldots, n_k)$ of $K^{n_1, \ldots, n_k}$.

We discuss a couple of preparatory results before giving the proof of \autoref{prop:gen-subrank-exists}. We define $X_r$ to be the set of tensors in $K^{n,n,n}$ whose $[r]\times[r]\times[r]$ subtensor is zero except for the diagonal entries in $[r]\times [r]\times [r]$ which are all nonzero,
\[
X_r = \left\{T \in K^{n,n,n}\ | \ T_{ijk} = 0 \text{ for } (i,j,k) \in [r]^3 \setminus \{(i,i,i)\ |\ i \in [r]\} \text{ and } T_{i,i,i} \neq 0 \text{ for } i \in [r] \right\}.
\]
We let $\psi_r$ be the map that applies basis transformations to elements of $X_r$,
\begin{align*}
\psi_r: \GL_n \times \GL_n \times \GL_n \times X_r &\longrightarrow K^{n,n,n}\\
(A,B,C,T) & \longmapsto (A \otimes B \otimes C)T.
\end{align*}
We define $\mathcal{C}_r$ to be the set of tensors in $K^{n,n,n}$ whose subrank is at least $r$,
\[
\mathcal{C}_r = \{T \in K^{n,n,n}\ | \ \subrank(T) \geq r\}.
\]

\begin{lemma} \label{lem:Cr-as-image}
The image of $\psi_r$ is precisely $\mathcal{C}_r$.
\end{lemma}

\begin{proof}
To prove $\im(\psi_r) = \mathcal{C}_r$ we show both inclusions.

First we will prove $\im(\psi_r) \subseteq \mathcal{C}_r$. As a first step we will prove that $X_r \subseteq \mathcal{C}_r$.
Let $T \in X_r$. Let $T_{i,i,i} = \lambda_i \neq 0$. Let 
\[
A = \begin{pmatrix} {\rm Id}_r & 0 \\ 0& 0 \end{pmatrix}
\]
where ${\rm Id}_r$ denotes the identity matrix of size $r \times r$ and let 
\[
B = \begin{pmatrix} D & 0 \\ 0 & 0\end{pmatrix}
\]
where $D$ is a diagonal matrix of size $r \times r$ whose diagonal entries are $\lambda_1^{-1},\lambda_2^{-1},\dots, \lambda_r^{-1}$. It is easy to check that $(A \otimes A \otimes B) \cdot T = I_r$. Thus $X_r \subseteq \mathcal{C}_r$. Since subrank is invariant under the action of $\GL_n \times \GL_n \times \GL_n$, we see that $\mathcal{C}_r$ is $\GL_n \times \GL_n \times \GL_n$ invariant, so we deduce that $\im(\psi_r) = (\GL_n \times \GL_n \times \GL_n) \cdot X_r \subseteq \mathcal{C}_r.$
 
Now we will prove $\mathcal{C}_r \subseteq \im(\psi_r)$. Let $T \in \mathcal{C}_r$. Then, there exist $A,B,C \in \Mat_{r,n}$ such that $(A \otimes B \otimes C) \cdot T = I_r$. Let \[
\widetilde{A} = \begin{pmatrix} A \\ \ast \end{pmatrix}
\]
be a completion of $A$ to a full rank $n \times n$ matrix, and similarly define $\widetilde{B}$ and $\widetilde{C}$. This is possible because $A,B,C$ must all have rank $r$. Then $(\widetilde{A} \otimes \widetilde{B} \otimes \widetilde{C})\cdot T \in X_r$, which implies that $T \in \im(\psi_r)$. Thus $\mathcal{C}_r \subseteq \im(\psi_r)$
\end{proof}

We now give the proof of the existence of the generic subrank (\autoref{prop:gen-subrank-exists}). The proof is standard and may safely be skipped.

\begin{proof}[{Proof of \autoref{prop:gen-subrank-exists}}]
The map $\subrank$ attains only finitely many values on $V = K^{n,n,n}$, namely $\{0,1,\ldots, n\}$. Thus it has finitely many fibers $P_i = \subrank^{-1}(i) \subseteq V$. Each $P_i$ is a constructible set, since $P_i = \mathcal{C}_i \setminus \mathcal{C}_{i+1}$ and $\mathcal{C}_i$ is constructible as a consequence of \autoref{lem:Cr-as-image}. Then $\bigcup_{i=0}^n P_i = V$ and so $\bigcup_{i=0}^n \overline{P_i} = V$. However, $V$ is irreducible so there must be an $i$ such that $\overline{P_i} = V$. Since~$P_i$ is constructible it contains a subset $U \subseteq P_i$ that is non-empty and Zariski-open in $\overline{P_i} = V$ (this is a general fact, see e.g.~\cite{MR1102012}). This $U$ and $i$ satisfy the claim.
\end{proof}

Now that we have established that the generic subrank exists, we continue to prove our upper bound on it.
The following simple lemma is straightforward, but crucial:
\begin{lemma} \label{lem:dimC-subrank}
The generic subrank $\gensubrank(n) = \max \{r\ |\ \dim(\mathcal{C}_r) = n^3\}$.
\end{lemma}

\begin{proof}
As above, let $P_i$ denote the subset of tensors with subrank $i$. Then, by definition of $\gensubrank(n)$ and the fact that $\mathcal{C}_r = \bigsqcup_{i =r}^n P_i$, we deduce that $\dim(\mathcal{C}_r) = n^3$ if and only if $\mathcal{C}_r \supseteq P_{\gensubrank(n)}$ if and only if $r \leq \gensubrank(n)$.
\end{proof}

\begin{proposition} \label{prop:dim-upper}
The following is an upper bound for the dimension of $\mathcal{C}_r$:
$$
\dim(\mathcal{C}_r) \leq 3n^2 + (n^3-r^3+r) - 3(n(n-r) + r) = n^3 - r(r^2-3n+2).
$$
\end{proposition}

\begin{proof} %
Consider the map $\psi$ above. The theorem on dimension of fibres says that 
$$
\dim(\mathcal{C}_r) = \dim(\GL_n \times \GL_n \times \GL_n \times X_r) - \dim(\text{general fiber of } \psi_r).
$$
We see that $\dim(\GL_n \times \GL_n \times \GL_n \times X_r) = 3n^2 + (n^3-r^3 + r)$ because each $\GL_n$ contributes~$n^2$ to the dimension and $X_r$ is a Zariski-open subset of a linear subspace of $K^{n,n,n}$ of dimension $(n^3-r^3 + r)$. Thus, to compute $\dim(\mathcal{C}_r)$, we only need to compute the dimension of the general fiber of $\psi_r$. Note that the dimension of any fiber is at most the dimension of the general fiber, so it suffices to find a lower bound on the dimension of all fibers.

Suppose $T \in \mathcal{C}_r$, then $T = (g_1 \otimes g_2 \otimes g_3) \cdot S$ for some $S \in X_r$ by \autoref{lem:Cr-as-image}. It is easy to see that $\dim(\psi^{-1}(T)) = \dim(\psi^{-1}(S))$ since $\psi^{-1}(T)$ and $\psi^{-1}(S)$ are isomorphic as varieties -- indeed this follows from the observation that $(A,B,C,U) \in \psi^{-1}(T) \Longleftrightarrow (g_1^{-1}A, g_2^{-1}B, g_3^{-1}C, U) \in \psi^{-1}(S)$. Thus, it suffices to lower bound the dimension of $\psi^{-1}(S)$ for $S \in X_r$.

Let $S \in X_r$.
Let $L \subseteq \GL_m$ be the subset of matrices of the form 
\[
\begin{pmatrix} D & 0 \\ \ast & \ast \end{pmatrix}
\]
where $D$ is a diagonal matrix of size $r \times r$. It is easy to see that $\dim(L) = n(n-r) + r$. For $A,B,C \in L$, it is easy to see that $(A \otimes B \otimes C) \cdot S \in X_r$. Thus, $(A^{-1}, B^{-1}, C^{-1}, (A \otimes B \otimes C) \cdot S) \in \psi^{-1}(S)$. In particular, this means that $\dim(\psi^{-1}(S)) \geq 3 \dim(L) = 3(n(n-r) + r)$. Thus, we conclude that $\dim(\text{generic fiber}) \geq 3(n(n-r) + r)$ and so
\[
\dim(\mathcal{C}_r) \leq 3n^2 + (n^3-r^3 + r) - 3(n(n-r) + r) = n^3 - r(r^2-3n+2).\qedhere
\]

\end{proof}

Now, we have everything necessary to prove the upper bound for the subrank of $3$-tensors.
\begin{proof} [Proof of \autoref{thm:subrank-upper-3-tensor}]
Suppose the subrank of a generic tensor in $K^{n,n,n}$ is $t = \gensubrank(n)$. Then we know that $n^3 = \dim(\mathcal{C}_t)$ by Lemma~\ref{lem:dimC-subrank} and $\dim(\mathcal{C}_t) \leq n^3 - t(t^2 - 3n+2)$ by  Proposition~\ref{prop:dim-upper}. Thus, we must have $t^2 - 3n + 2 \leq 0$, so $t \leq \sqrt{3n-2}$. Hence, we have $\gensubrank(n) \leq \sqrt{3n-2}$ as desired. 
\end{proof}

\subsection{Higher-order tensors}
In the general case, where we look at tensors in $K^{n_1,n_2,\dots,n_k}$, we define analogously the objects~$X_r$ and $\mathcal{C}_r$, the map $\psi: \prod_{i=1}^k \GL_{n_i} \times X_r \rightarrow \mathcal{C}_r$, etc. 

\begin{proof} [Proof of \autoref{thm:subrank-upper-k-tensor}]
We obtain analogously that 
$$
\dim(\mathcal{C}_r) = \Bigl(\sum_{i=1}^k n_i^2\Bigr) + \Bigl(\prod_{i=1}^k n_i - r^k + r\Bigr) - \sum_{i=1}^k (n_i(n_i - r) + r) = \prod_{i=1}^k n_i - r\Bigl(r^{k-1} - \sum_{i=1}^k n_i + (k-1)\Bigr).
$$
Suppose the subrank of a generic tensor in $K^{n_1,n_2,\dots,n_k}$ is $t = \gensubrank(n_1,\ldots,n_k)$. Then, we have
$$
\prod_{i=1}^k n_i = \dim(\mathcal{C}_t) \leq \prod_{i=1}^k n_i - t\Bigl(t^{k-1} - \sum_{i=1}^k n_i + (k-1)\Bigr),
$$
so we get that $t^{k-1} - \sum_{i=1}^k n_i + (k-1) \leq 0$, so that 
\[
\gensubrank(n) = t \leq \Bigl(\sum_{\smash{i=1}}^k n_i - (k-1)\Bigr)^{\frac{1}{k-1}}
\]
as desired.
\end{proof}

\section{Lower bounds on generic subrank}\label{sec:lower-bounds}

In this section, we describe the technique we use to show lower bounds on generic subrank. In order to make this technique effective, we will need some explicit constructions of linear subspaces which we postpone to the next section. 

First, we introduce some notation. We identify $\Mat_{n_2,n_3}$ with $K^{n_2} \otimes K^{n_3}$ in the standard way. For a linear subspace $\mathcal{X} \subseteq \Mat_{n_2,n_3}$, we define a linear subspace
$$
\mathcal{X}[1] = K^{n_1} \otimes \mathcal{X} \subseteq K^{n_1,n_2,n_3}.
$$
The linear subspace $\mathcal{X}[1]$ consists precisely of the tensors whose slices in the first direction (i.e., matrices $(T_{1jk})_{j,k},(T_{2jk})_{j,k},\dots,(T_{njk})_{j,k}$) are in $\mathcal{X}$. Similarly, we define $\mathcal{X}[2]$ (resp. $\mathcal{X}[3]$)  as the set of tensors whose slices in the second (resp. third) direction are in $\mathcal{X}$.

The main idea behind proving our lower bounds is the following result:

\begin{theorem} \label{thm:technical-lower-bound}
Let $r \leq n_1,n_2,n_3$ such that $n_i - r \leq r^2$.\footnote{Without this assumption it does not mean anything to have a generic subspace of dimension $n_i - r$. Moreover, if $n_i - r  > r^2$, then it is easy to see that $\gensubrank(n_1,n_2,n_3) \geq r$.} Let $\mathcal{X}_i \subseteq \Mat_{r,r}$ be a generic subspace of dimension $n_i - r$ for $i = 1,2,3$. Suppose $\mathcal{X}_1[1] + \mathcal{X}_2[2] + \mathcal{X}_3[3] = K^{r,r,r}$, then
$$
\gensubrank(n_1,n_2,n_3) \geq r.
$$
\end{theorem}

In fact, we can prove a stronger version of the above theorem, which we state as \autoref{thm:technical-lower-bound-strong}. However, the downside of this stronger version is that it has a hypothesis that is more difficult to work with. 

To prove that $r$ is a lower bound for the generic subrank, we need to show that the image of~$\psi_r$ has dimension $n^3$, that is, the image is Zariski-dense in $K^{n,n,n}$. The dimension of the image of a map can be captured by the rank of the differential at a generic point -- an idea that is familiar to differential geometers and algebraic geometers alike. In arbitrary characteristic, we know that the rank of the differential at a generic point is a lower bound for the dimension of the image, which is sufficient for proving lower bounds. In characteristic $0$, the image of the rank of the differential at a generic point is equal to the dimension of the image. Consequently, we are able to obtain an exact linear algebraic characterization for generic subrank in \autoref{thm:technical-lower-bound-strong}.

Recall that we defined the map 
$$
\psi_r : \GL_{n_1} \times \GL_{n_2} \times \GL_{n_3} \times X_r \rightarrow K^{n_1,n_2,n_3},
$$
where 
$$
X_r = \left\{T \in K^{n_1,n_2,n_3} \ | \ T_{ijk} = 0 \text{ for } (i,j,k) \in [r]^3 \setminus \{(i,i,i)\ |\ i \in [r]\} \text{ and } T_{i,i,i} \neq 0 \text{ for } i \in [r] \right\}.
$$
Observe that 
\[
Y_r = \left\{T \in K^{n_1,n_2,n_3} \ | \ T_{ijk} = 0 \text{ for } (i,j,k) \in [r]^3 \setminus \{(i,i,i)\ |\ i \in [r] \} \right\}
\]is the tangent space of $X_r$. The differential at a point $(g_1,g_2,g_3,T)$ is
$$
 (d\psi_r)_{(g_1,g_2,g_3,T)} : \Mat_{n_1,n_1} \times \Mat_{n_2,n_2} \times \Mat_{n_3,n_3} \times Y_r  \longrightarrow  K^{n_1,n_2,n_3} 
$$
given by 
$$
(A,B,C,S)  \longmapsto  ((A \otimes g_2 \otimes g_3) + (g_1 \otimes B \otimes g_3) + (g_1 \otimes g_2 \otimes C))T + (g_1 \otimes g_2 \otimes g_3)S.
$$

\begin{lemma}
If for generic $(g_1,g_2,g_3,T) \in \GL_{n_1} \times \GL_{n_2} \times \GL_{n_3} \times X_r$ we have
\[
\im((d\psi_r)_{(g_1,g_2,g_3,T)}) = K^{n_1,n_2,n_3},
\]
then $\gensubrank(n_1,n_2,n_3) \geq r$. The converse holds if the characteristic of $K$ is $0$.
\end{lemma}

\begin{proof}
If the rank of the differential $d\psi_r$ at a generic point is full, then the image of $\psi_r$ must be full dimensional, that is, Zariski-dense (and constructible). Every tensor in the image of $\psi_r$ has subrank $\geq r$. Hence, the generic subrank must be at least $r$.

Assume now that characteristic of $K$ is zero. If the rank of $d \psi_r$ at a generic point is not full, then the image of $\psi_r$ is not full dimensional, that is, the set of tensors having subrank $\geq r$ is not Zariski-dense, so we get $\gensubrank(n_1,n_2,n_3) < r$, thereby proving the converse.
\end{proof}

We use the following equivariance property of the differential.

\begin{lemma}
We have for all $g_i \in \GL_{n_i}$ and $T \in X_r$ that
\[
\im( (d\psi_r)_{(g_1,g_2,g_3,T)}) = (g_1 \otimes g_2 \otimes g_3) (\im( (d\psi_r)_{(I,I,I,T)}).
\]
\end{lemma}

\begin{proof}
We see directly that
\[
(d\psi_r)_{(g_1,g_2,g_3,T)} (A,B,C,S) =  (g_1 \otimes g_2 \otimes g_3) (d\psi_r)_{(I,I,I,T)} (g_1^{-1} A, g_2^{-1}B, g_3^{-1}C, S)
\]
which implies the claim.
\end{proof}

\begin{corollary} \label{cor:im-d-subrank}
If for generic $T \in X_r$ we have $
\im( (d\psi_r)_{(I,I,I,T)}) = K^{n_1,n_2,n_3}$, then we have that $\gensubrank(n_1,n_2,n_3) \geq r$. The converse holds if the characteristic of $K$ is $0$.
\end{corollary}

So, let us now analyze more carefully the image of $(d\psi_r)_{(I,I,I,T)}$ for a generic tensor $T$. Below, we write $d$ for $(d\psi_r)_{(I,I,I,T)}$ for notational simplicity.
Thus $d$ is given by
\[
d(A,B,C,S) = ((A \otimes I \otimes I) + (I \otimes B \otimes I) + (I \otimes I \otimes C)) \cdot T + S.
\]
The map $d$ is a linear map, so we see that
\[
\im\  d = \sum_{i=1}^3 d(\Mat_{n_i,n_i}) + d(Y_r)
\]
where we use the notation $d(Y_r) = d(0,0,0,Y_r)$,  $d(\Mat_{n_1, n_1}) = d(\Mat_{n_1, n_1}, 0,0,0)$, etc.

\begin{lemma} \label{lem:diff-image-Y_r}
The image of $Y_r$ under the differential map $d$ is 
$$
d(Y_r) = Y_r
$$
\end{lemma}
\begin{proof}
Observe that for any $S \in Y_r$, we have $d(0,0,0,S) = S$. 
\end{proof}

Let $L_i = [T_{ijk}]_{j,k} $ for $1 \leq i \leq n_1$. These just split the tensor in the first direction as a stack of $n$ matrices. Let $\mathcal{L}$ denote the span of the $L_i$s.

\begin{lemma} \label{lem:diff-image-1st}
The image of $\Mat_{n_1,n_1}$ under the differential $d$ is
\begin{equation} \label{eq:image-differential-linear}
d(\Mat_{n_1,n_1}) = \mathcal{L}[1].
\end{equation}
\end{lemma}

\begin{proof}
For any $A \in \Mat_{n_1,n_1}$, let $T' = d(A,0,0,0) = (A \otimes I \otimes I) \cdot T$. Consider the slices $L_i = [T_{ijk}]_{j,k} $. Then, it is straighforward to compute that for $1 \leq i \leq n_1$, the slice
\[
[T'_{ijk}]_{j,k}= \sum_{t = 1}^{n_1} a_{it} L_i.\qedhere
\]
\end{proof}

Similar to $L_i$, define the slices in the other directions as $M_j = [T_{ijk}]_{i,k}$ and $N_k = [T_{ijk}]_{i,j}$. Let $\mathcal{M}$ and $\mathcal{N}$ denote the spans of the $M_j$s and $N_k$s respectively.

\begin{corollary} \label{cor:image-d}
The image of the differential $d$ is
$$
\im\ d = \mathcal{L}[1] + \mathcal{M}[2] + \mathcal{N}[3] + Y_r.
$$
\end{corollary}

\begin{proof}
This follows from %
Lemma~\ref{lem:diff-image-1st} (and its counterparts in the other two directions) and Lemma~\ref{lem:diff-image-Y_r}.
\end{proof}

Now, we refine the above corollary. Denote by $\widehat{L_i}$ the top-left $ r \times r$ submatrix of $L_i$. Then, let $\mathcal{\widehat{L}} = \spa (\widehat{L_1}, \widehat{L_2} ,\dots, \widehat{L_{n_1}})$. Similarly define $\widehat{M_j}, \widehat{N_k}$ and $\mathcal{\widehat{M}}$ and $\mathcal{\widehat{N}}$.

\begin{proposition} \label{prop:im-d}
The image of the differential $d$ is
$$
\im\ d = \mathcal{\widehat{L}}[1] + \mathcal{\widehat{M}}[2] + \mathcal{\widehat{N}}[3] + Y_r.
$$
\end{proposition}

\begin{proof}
This follows directly from Corollary~\ref{cor:image-d} and the definition of the coordinate subspace~$Y_r$.
\end{proof}

Now, we can finally prove \autoref{thm:technical-lower-bound}.

\begin{proof} [Proof of \autoref{thm:technical-lower-bound}]
Let $T \in X_r$ be generic. Define $L_i, M_i, N_i$ as above. Let
\[
\mathcal{X}_1 = \spa (\widehat{L_{r+1}}, \widehat{L_{r+2}}, \dots, \widehat{L_{n_1}}).
\]
Then $\mathcal{X}_1$ is a generic subspace of $\Mat_{r,r}$ of dimension $n_1 - r$. We similarly define 
\begin{align*}
\mathcal{X}_2 &= \spa (\widehat{M_{r+1}}, \widehat{M_{r+2}}, \dots, \widehat{M_{n_2}})\\
\mathcal{X}_3 &= \spa (\widehat{N_{r+1}}, \widehat{N_{r+2}}, \dots, \widehat{N_{n_3}}).
\end{align*}
By hypothesis, $\mathcal{X}_1[1] + \mathcal{X}_2[2] + \mathcal{X}_3[3] = K^{r,r,r}$. Hence, by Proposition~\ref{prop:im-d}, we see that 
\[
\im(d) = \mathcal{\widehat{L}}[1] + \mathcal{\widehat{M}}[2] + \mathcal{\widehat{N}}[3] + Y_r \supseteq \mathcal{X}_1[1] + \mathcal{X}_2[2] + \mathcal{X}_3[3] + Y_r = K^{r,r,r} + Y_r = K^{n_1,n_2,n_3}.
\]
Thus, we get that $\gensubrank(n_1,n_2,n_3) \geq r$  by Corollary~\ref{cor:im-d-subrank}.
\end{proof}

We observe that the idea from the proof of \autoref{thm:technical-lower-bound} actually yields the stronger version below. To state this we define the linear subspace
\[
W_r =  \{T \in K^{r,r,r} \ | \ T_{ijk} = 0 \text{ if  $i,j,k$ are all different} \} \subseteq K^{r,r,r}.
\]\par%
\begin{theorem}\label{thm:technical-lower-bound-strong}
Let $r \leq n_1,n_2,n_3$ such that $n_i = r < r^2$ for $i = 1,2,3$. Let $\mathcal{X}_i \subseteq \Mat_{r,r}$ be a generic subspace of dimension $n_i - r$ for $i = 1,2,3$. Suppose 
$$
\mathcal{X}_1[1] + \mathcal{X}_2[2] + \mathcal{X}_3[3] + W_r = K^{r,r,r},
$$
then $\gensubrank(n_1,n_2,n_3) \geq r$.
Further, if characteristic of $K$ is $0$, then we have the converse.
\end{theorem}
In other words, if characteristic of $K$ is $0$, then \autoref{thm:technical-lower-bound-strong} says that $\gensubrank(n_1,n_2,n_3)$ is given by the smallest number $r$ such that 
$\mathcal{X}_1[1] + \mathcal{X}_2[2] + \mathcal{X}_3[3] + W_r = K^{r,r,r}$ for generic subspaces $\mathcal{X}_i \subseteq \Mat_{r,r}$ of dimension $n_i -r$.

\begin{proof}
The proof is similar to that of \autoref{thm:technical-lower-bound}.
Take a generic $T \in X_r$. Define $L_i, M_i, N_i$ as above. Set $\mathcal{X}_1 = \spa (\widehat{L_{r+1}}, \widehat{L_{r+2}}, \dots, \widehat{L_{n_1}})$. Then $\mathcal{X}_1$ is a generic subspace of $\Mat_{r,r}$ of dimension $n_1 - r$. Let $P_1 = \spa (\widehat{L_1}, \widehat{L_2},\dots,\widehat{L_r})$. Similarly define $\mathcal{X}_2 = \spa (\widehat{M_{r+1}}, \widehat{M_{r+2}}, \dots, \widehat{M_{n_2}})$ and $\mathcal{X}_3 = \spa (\widehat{N_{r+1}}, \widehat{N_{r+2}}, \dots, \widehat{N_{n_3}})$ and also $P_2 = \spa (\widehat{M_1},\dots,\widehat{M_r})$ and $P_3 = \spa (\widehat{N_1}, \dots, \widehat{N_r})$.

It is a straightforward computation to see that since $T$ is generic in $X_r$, we have $P_1[1] + P_2[2] + P_3[3] = W_r$.
Hence, using Proposition~\ref{prop:im-d}, we see that 
\begin{align*}
\im(d) &= \mathcal{\widehat{L}}[1] + \mathcal{\widehat{M}}[2] + \mathcal{\widehat{N}}[3] + Y_r \\
&= \mathcal{X}_1[1] + \mathcal{X}_2[2] + \mathcal{X}_3[3] + P_1[1] + P_2[2] + P_3[3] +  Y_r  \\
& = \mathcal{X}_1[1] + \mathcal{X}_2[2] + \mathcal{X}_3[3] + W_r + Y_r
\end{align*}

Now, we claim that 
$$
\mathcal{X}_1[1] + \mathcal{X}_2[2] + \mathcal{X}_3[3] + W_r = K^{r,r,r} \Leftrightarrow \mathcal{X}_1[1] + \mathcal{X}_2[2] + \mathcal{X}_3[3] + W_r + Y_r = K^{n_1,n_2,n_3}.
$$
The proof of the claim is rather straightforward, the only subtle point being that $Y_r$ and $K^{r,r,r}$ have an intersection (i.e., the main diagonal of $K^{r,r,r}$), but this intersection is included in $W_r$ anyway, so the claim goes through.

Thus, we get that $\im(d) = K^{n_1,n_2,n_3}$ if and only if $\mathcal{X}_1[1] + \mathcal{X}_2[2] + \mathcal{X}_3[3] + W_r = K^{r,r,r}$. Now, the theorem follows from applying Corollary~\ref{cor:im-d-subrank}.
\end{proof}

In the next section we will prove that the requirement of \autoref{thm:technical-lower-bound} is satisfied. This goes as follows. We will show (\autoref{lem:333-333}) that:

\begin{lemma}\label{lem:333}
Let $\mathcal{X}_i \subseteq \Mat_{3,3}$ be a generic subspace of dimension 3 for $i=1,2,3$. Then $\mathcal{X}_1[1]+\mathcal{X}_2[2]+\mathcal{X}_3[3] = K^{3,3,3}$.
\end{lemma}

Then from a blow-up argument we obtain:

\begin{lemma}\label{lem:333-blowup}
Let $\mathcal{X}_i \subseteq \Mat_{3d,3d}$ be a generic subspace of dimension $3d^2$ for $i=1,2,3$. Then $\mathcal{X}_1[1]+\mathcal{X}_2[2]+\mathcal{X}_3[3] = K^{3d, 3d, 3d}$.
\end{lemma}
\begin{proof}
It is enough to construct one such choice of $\mathcal{X}_i$. By \autoref{lem:333} there exist subspaces $\mathcal{Y}_i \subseteq \Mat_{3,3}$ of dimension 3 such that $\mathcal{Y}_1[1]+\mathcal{Y}_2[2]+\mathcal{Y}_3[3] = K^{3,3,3}$. Then $\mathcal{X}_i = \mathcal{Y}_i \otimes \Mat_{d,d}$ satisfy the claim.
\end{proof}

\begin{theorem}\label{thm:intermediate}
Let $d$ be a positive integer such that $n-3d \geq 3d^2$. Then $\gensubrank(n,n,n) \geq 3d$.
\end{theorem}
\begin{proof}
Taking $r = 3d$, this follows from \autoref{thm:technical-lower-bound} and \autoref{lem:333-blowup}.
\end{proof}

\begin{proof}[Proof of \autoref{thm:subrank-lower-3-tensor}]
If we take $d = \lfloor \sqrt{ n/3 + 1/4}   - 1/2 \rfloor$, then we have 
$d \leq \sqrt{n/3 + 1/4} - 1/2$, so $(d+1/2)^2 \leq n/3 + 1/4$. Thus, $d^2 + d + 1/4 \leq n/3 + 1/4$, so $d^2 + d \leq n/3$ or equivalently $3d^2 \leq n-3d$. Thus, it follows from \autoref{thm:intermediate} that $\gensubrank(n) \geq 3d.$
\end{proof}

For the higher-order lower bound \autoref{thm:subrank-lower-k-tensor}
we prove (\autoref{lem:higher-order-construction}) the analogous higher-order version of \autoref{lem:333} which we can similarly blow up and apply \autoref{thm:technical-lower-bound} to.

\section{Constructions of tensor space decompositions}\label{sec:constructions}

\newcommand{\kar} {{\rm char}}
\newcommand{\R}{{\mathbb R}}
\newcommand{\Rep}{{\rm Rep}}
\newcommand{\SI}{{\rm SI}}
\newcommand{\SO}{{\rm SO}}
\newcommand{\ST}{{\rm ST}}
\newcommand{\PD} {{\rm PD}}
\newcommand{\mlt}{{\rm mlt}}
\newcommand{\Lie} {{\rm Lie}}

\newcommand{\B}{{\mathcal B}}
\newcommand{\C}{{\mathbb C}}
\newcommand{\N}{{\mathbb N}}
\newcommand{\F}{{\mathbb F}}
\newcommand{\Q}{{\mathbb Q}}
\newcommand{\Z}{{\mathbb Z}}
\newcommand{\tnr}{\operatorname{Tr}}
\newcommand{\tnrace}{\operatorname{Tr}}
\newcommand{\SL}{\operatorname{SL}}
\newcommand{\Mod}{\operatorname{mod}}

\newcommand{\VV}{{\mathbb V}}

\newcommand{\nullcone}{\mathcal{N}}
\newcommand{\Sym}{S}

\newcommand{\rk} {\operatorname{rk}}
\newcommand{\g} {\mathfrak{g}}
\newcommand{\n} {\mathfrak{n}}
\newcommand{\h} {\mathfrak{h}}
\newcommand{\tn}{\mathfrak{t}}
\newcommand{\sln}{\mathfrak{sl}_n}
\newcommand{\adaptable} {adaptable}
\newcommand{\End}{{\rm End}}
\newcommand{\U}{\operatorname{U}}
\newcommand{\SU}{\operatorname{SU}}
\newcommand{\ad}{\rm ad}
\newcommand{\uncramped} {uncramped}
\newcommand{\Supp} {{\rm Supp}}
\newcommand{\wt}{\widetilde}
\newcommand{\isom} {\stackrel{\sim}{\longrightarrow}}
 \newcommand{\sslash}{\mathbin{\mkern-4mu/\mkern-6mu/\mkern-4mu}}
 \newcommand{\CC}{{\mathcal C}}
\newcommand{\SSS}{{\mathcal S}}

\newcommand{\comm}[1]{\textcolor{red}{$\clubsuit${\it #1}$\clubsuit$}}

We assume that the base field $K$ is infinite. 

Let $\CC$ be the set of all matrices $\begin{bsmallmatrix} n_1 & n_2 & \cdots & n_d\\a_1 & a_2  & \cdots & a_d\end{bsmallmatrix}$  for which  $n_1,n_2,\dots,n_d$ are positive integers, $a_1,a_2,\dots,a_d$ are nonnegative integers
and $\sum_{i=1}^d a_i n_i=\prod_{i=1}^d n_i$. Let $\SSS$ be the subset of all  $\begin{bsmallmatrix} n_1 & n_2 & \cdots & n_d\\a_1 & a_2  & \cdots & a_d\end{bsmallmatrix}\in {\CC}$ with the following property:
If $V_1,V_2,\dots,V_d$ are vector spaces of dimensions $n_1,n_2,\dots,n_d$ respectively, and $W_i$ is a general subspace
of $\widehat{V}_i:=V_1\otimes V_2\otimes \cdots \otimes V_{i-1}\otimes V_{i+1}\otimes \cdots\otimes V_d$ of dimension $a_i$ for all~$i$,
then $\sum_{i=1}^d \Phi_i( W_i\otimes V_i) = V_1 \otimes \cdots \otimes V_d$, where $\Phi:\widehat{V}_i\otimes V_i\to V_1\otimes V_2\otimes \cdots \otimes  V_d$ is
the isomorphism given by permuting the factors.

\subsection{General construction methods}

We will use the following simple facts. It is obvious that if a matrix lies in $\SSS$ and we permute its columns, then it will still lie in $\SSS$.
It is also clear that  $\begin{bsmallmatrix} n_1 & n_2 & \cdots & n_d\\a_1 & a_2  & \cdots & a_d\end{bsmallmatrix}$  lies in $\SSS$ if and only
if  $\begin{bsmallmatrix} n_1 & n_2 & \cdots & n_d & 1\\a_1 & a_2  & \cdots & a_d & 0\end{bsmallmatrix}$ lies in $\SSS$. 
The vector $\begin{bsmallmatrix}n\\1\end{bsmallmatrix}$ lies in~$\SSS$ for all positive integers $n$. Finally,
\begin{equation}\label{eq:refine}
\begin{bsmallmatrix} n_1 & n_2 & \cdots & n_d & 1 & 1 \\a_1 & a_2  & \cdots & a_d & b_1 & b_2\end{bsmallmatrix} \in \SSS \textnormal{\quad if and only if \quad}  \begin{bsmallmatrix} n_1 & n_2 & \cdots & n_d & 1 \\a_1 & a_2  & \cdots & a_d & b_1 + b_2\end{bsmallmatrix} \in \SSS
\end{equation}

\begin{lemma}[Direct sum construction]\label{lem:direct-sum-construction}
Suppose that $n_d=n_d'+n_d''$, and $a_i=a_i'+a_i''$ for $i=1,2,\dots,d-1$, and
$$
\begin{bsmallmatrix} n_1 & n_2 & \dots & n_{d-1} & n_d'\\
a_1' & a_2' & \cdots & a_{d-1}' & a_d\end{bsmallmatrix},\begin{bsmallmatrix} n_1 & n_2 & \dots & n_{d-1} & n_d''\\
a_1'' & a_2'' & \cdots & a_{d-1}'' & a_d\end{bsmallmatrix} \in \SSS.
$$
Then we have $\begin{bsmallmatrix} n_1 & n_2 & \cdots & n_d\\a_1 & a_2  & \cdots & a_d\end{bsmallmatrix}\in \SSS$.
\end{lemma}
\begin{proof}
Suppose $V_1,V_2,\dots,V_d$ are vector spaces of dimensions $n_1,n_2,\dots,n_d$ respectively and choose
a general subspace $W_d$ of $V_d$ of dimension $a_d$.
We can write $V_d=V_{d}'\oplus V_d''$ where $V_d'$ and $V_d''$ have dimensions $n_d'$ and $n_d''$ respectively.
For $i=1,2,\dots,d-1$ choose a general subspace $W_i'\subseteq V_1\otimes \cdots \otimes V_{i-1}\otimes V_{i+1}\otimes \cdots\otimes V_{d-1}\otimes V_d'$ of dimension $a_i'$ such that $V_1\otimes V_2\otimes \cdots \otimes V_{d-1}\otimes V_d'$ is equal to
$\big(\sum_{i=1}^{d-1}\Phi_i( W_i'\otimes V_i)\big)+W_d\otimes V_d'$.
Similarly, for $i=1,2,\dots,d-1$,   choose a general subspace $W_i''\subseteq V_1\otimes \cdots \otimes V_{i-1}\otimes V_{i+1}\otimes \cdots\otimes V_{d-1}\otimes V_d''$ of dimension $a_i''$ such that $V_1\otimes V_2\otimes \cdots \otimes V_{d-1}\otimes V_d''$ is equal to
$\big(\sum_{i=1}^{d-1}\Phi_i( W_i''\otimes {V_i})\big)+W_d\otimes V_d''$.
Set $W_i=W_i'\oplus W_i''\subseteq V_1\otimes \cdots\otimes V_{i-1}\otimes V_{i+1}\otimes \cdots \otimes V_{d}$ for $i=1,2,\dots,d-1$.
Then we have 
\begin{multline*}
\sum_{i=1}^d \Phi_i(W_i\otimes V_i)=\Big(\sum_{i=1}^{d-1} \Phi_i((W_i'\oplus W_i'')\otimes V_i)\Big)+W_d\otimes (V_d'\oplus V_d'')=\\
\left(\Big(\sum_{i=1}^{d-1} \Phi_i(W_i'\otimes V_i)\Big)+W_d\otimes V_d'\right)\oplus \left(\Big(\sum_{i=1}^{d-1} \Phi_i(W_i''\otimes V_i)\Big)+W_d\otimes V_d''
\right)=\\
(V_1\otimes \cdots \otimes V_{d-1}\otimes V_d')\oplus (V_1\otimes \cdots \otimes V_{d-1}\otimes V_d'')=
V_1\otimes V_2\otimes \cdots \otimes V_d.
\end{multline*}
This finishes the proof.
\end{proof}

\subsection{Recursive construction for every order}

The following lemma gives \autoref{lem:333} which we used in the proof of the generic subrank lower bound for order three. After that we will recursively obtain what is needed for the higher-order case generic subrank lower bound.

\begin{lemma}[Base case]\label{lem:333-333}
We have $\begin{bsmallmatrix} 3& 3 &3\\3 & 3 &3\end{bsmallmatrix}\in \SSS$.
\end{lemma}
\begin{proof}
We can verify explicitly that $\begin{bsmallmatrix} 2 & 2 &2\\ 0 & 3 & 1\end{bsmallmatrix}\in \SSS$, as follows.
For a $3$-dimensional subspace $W_1$ of $V_2\otimes V_3=K^{2\times 2}$, take the space of all matrices of the form
$$
\begin{pmatrix}
a & b\\
b & c
\end{pmatrix}
$$
and for the $1$-dimensional subspace $W_3$ of $V_1\otimes V_2=K^{2\times 2}$, take the span of the identity matrix. 
A tensor in the intersection of $W_1\otimes V_3\cap V_1\otimes W_3\subseteq V_1\otimes V_2\otimes V_3$ is of the form
$$
\left(
\begin{array}{cc|cc}
a_1 & b_1 & a_2 & b_2\\
b_1 & c_1 & b_2 & c_2
\end{array}
\right)= \left(\begin{array}{cc|cc}
p_1 & 0  & 0  & p_1\\
p_2 & 0  & 0 & p_2
\end{array}\right)
$$
for some $a_1,a_2,b_1,b_2,c_1,c_2,p_1,p_2$. Clearly, $b_1=c_1=a_2=b_2=0$. So we get $p_1=b_2=0$ and $p_2=b_1=0$.
So $W_1\otimes V_3\cap V_1\otimes W_3=0$ and $W_1\otimes V_3+V_1\otimes W_3$ has dimension
$3\cdot 2+2\cdot 1=8$ and therefore must be equal to $V_1\otimes V_2\otimes V_3=K^{2\times 2\times 2}$. Using \autoref{lem:direct-sum-construction},
we now deduce:
\begin{itemize}
\item $\begin{bsmallmatrix} 2 & 2 &2\\ 0 & 3 &1\end{bsmallmatrix}, \begin{bsmallmatrix} 1 & 2 & 2\\ 0 & 0 &2\end{bsmallmatrix}\in \SSS$, so
$\begin{bsmallmatrix} 3 & 2 & 2\\ 0 & 3 &3\end{bsmallmatrix}=\begin{bsmallmatrix} 2+1 & 2 & 2\\ 0 & 3+0 & 1 +2\end{bsmallmatrix}\in \SSS$
\item $\begin{bsmallmatrix} 3 & 1 & 1\\ 1 & 0& 0\end{bsmallmatrix},\begin{bsmallmatrix} 3 & 1 & 1\\ 0 & 3 & 0\end{bsmallmatrix}\in \SSS$, 
so $\begin{bsmallmatrix} 3 & 1 & 2\\ 1 & 3 &0\end{bsmallmatrix}=\begin{bsmallmatrix} 3 & 1 &1+1\\ 1+0 & 0 +3 & 0\end{bsmallmatrix}\in \SSS$
\item $\begin{bsmallmatrix} 3 & 2 & 2\\ 0 & 3 &3\end{bsmallmatrix},\begin{bsmallmatrix} 3 & 1 & 2\\ 1 & 3 &0\end{bsmallmatrix}\in \SSS$, so 
$\begin{bsmallmatrix} 3 & 3& 2\\ 1 & 3 &3\end{bsmallmatrix}=\begin{bsmallmatrix} 3 & 2+1 & 2\\ 0+1 & 3 &3+0\end{bsmallmatrix}\in \SSS$
\item $\begin{bsmallmatrix} 3 & 1 & 1\\1 & 0 & 0\end{bsmallmatrix}, \begin{bsmallmatrix} 3 & 1 & 1\\0 & 0 & 3\end{bsmallmatrix}\in \SSS$, so
 $\begin{bsmallmatrix} 3& 3 &1\\2 & 0 &3\end{bsmallmatrix}=\begin{bsmallmatrix} 3& 1+1+1 &1\\1+1+0 & 0 &0+0+3\end{bsmallmatrix}\in \SSS$
\item $\begin{bsmallmatrix} 3& 3 &1\\2 & 0 &3\end{bsmallmatrix}, \begin{bsmallmatrix} 3 & 3& 2\\ 1 & 3 &3\end{bsmallmatrix}\in \SSS$, so
$\begin{bsmallmatrix} 3 & 3& 3\\ 3 & 3 &3\end{bsmallmatrix}=\begin{bsmallmatrix} 3 & 3& 1+2\\ 2+1 & 0+3 &3\end{bsmallmatrix}\in \SSS$
\end{itemize}
The final line finishes the proof.
\end{proof}

To continue, we define the ``concatenation'' notation:
$$
\begin{bsmallmatrix} n_1 & n_2 & \cdots & n_d\\
a_1 & a_2 & \cdots & a_d\end{bsmallmatrix}\odot \begin{bsmallmatrix} m_1 & m_2 & \cdots & m_e\\
b_1 & b_2 & \cdots & b_e\end{bsmallmatrix}=
\begin{bsmallmatrix} n_1 & n_2 & \cdots  &n_d & m_1 & m_2 & \cdots & m_e\\
a_1 & a_2 & \cdots & a_d & b_1 & b_2 & \cdots & b_e\end{bsmallmatrix}
$$
and the $k$-fold repeated version of this notation:
$$
\begin{bsmallmatrix} n_1 & n_2 & \cdots & n_d\\
a_1 & a_2 & \cdots & a_d\end{bsmallmatrix}^{\odot k}=\underbrace{\begin{bsmallmatrix} n_1 & n_2 & \cdots & n_d\\
a_1 & a_2 & \cdots & a_d\end{bsmallmatrix}\odot \begin{bsmallmatrix} n_1 & n_2 & \cdots & n_d\\
a_1 & a_2 & \cdots & a_d\end{bsmallmatrix}\odot \cdots\odot \begin{bsmallmatrix} n_1 & n_2 & \cdots & n_d\\
a_1 & a_2 & \cdots & a_d\end{bsmallmatrix}}_k.$$
We will now use \autoref{lem:direct-sum-construction} and \autoref{lem:333-333} to prove:
\begin{lemma}\label{lem:higher-order-construction}
For all $n\geq 3$, we have $\begin{bsmallmatrix} n \\ n^{n-2}\end{bsmallmatrix}^{\odot n} \in \SSS$.
\end{lemma}
\begin{proof}
We already know the case $n=3$ (\autoref{lem:333-333}), so assume that $n\geq 4$.

We will first show that $\begin{bsmallmatrix} n \\ n \end{bsmallmatrix}^{\odot 3}\odot \begin{bsmallmatrix} 1 \\ n^2\end{bsmallmatrix}^{\odot (n-3)}\in \SSS$.
We note that a straightforward direct construction gives 
\[
  \begin{bsmallmatrix} n & n & 1\\ k & 0 & n^2-kn\end{bsmallmatrix}\in \SSS
\]
for $k=1,2,3$.
To proceed, we choose nonnegative integers $a_1,a_2,\dots,a_n,b_1,b_2,\dots,b_n\in \{0,1,2,3\}$
such that $a_1+a_2+\cdots+a_n=n$, $b_1+b_2+\cdots +b_n=n$ and $a_ib_i=0$ for all $i$.
Indeed this can be done as follows: If $n=2m$ is even, then we can take $a_1=a_2=\cdots=a_m=b_{m+1}=b_{m+2}=\cdots=b_{n}=2$
and $b_1=b_2=\cdots=b_m=a_{m+1}=a_{m+2}=\cdots=a_n=0$. If $n=2m+1$ is odd,
then we take $a_1=3$, $b_{m+1}=1$, $a_2=a_3=\cdots=a_m=b_{m+2}=b_{m+3}=\cdots=b_{n}=2$
and $b_1=b_2=\cdots=b_m=a_{m+1}=a_{m+2}=\cdots=a_n=0$.
Because $a_ib_i=0$ we get that 
\[
  \begin{bsmallmatrix} n & n & 1 \\ a_i & b_i  & n^2-(a_i+b_i)n\end{bsmallmatrix}\in \SSS
\]
and then using \eqref{eq:refine} we get
\begin{equation}\label{eq:elts-in-S}
  \begin{bsmallmatrix} n & n & 1 & 1\\ a_i & b_i  & n & n^2-(a_i+b_i+1)n\end{bsmallmatrix}\in \SSS
\end{equation}
for $i=1,2,\dots,n$.
Applying the direct sum construction \autoref{lem:direct-sum-construction} to the elements of \eqref{eq:elts-in-S} we get
\[
  \begin{bsmallmatrix} n & n &  n & 1\\ n & n & n & n^3-3n^2\end{bsmallmatrix}= \begin{bsmallmatrix} n & n &  1+\cdots+1 & 1\\
\sum_i a_i  & \sum_i b_i & n & \sum_i n^2-(a_i+b_i + 1)n\end{bsmallmatrix}\in \SSS.
\]
Applying \eqref{eq:refine} repeatedly to this we obtain
$\begin{bsmallmatrix} n \\ n \end{bsmallmatrix}^{\odot 3}\odot \begin{bsmallmatrix} 1 \\ n^2\end{bsmallmatrix}^{\odot (n-3)}\in \SSS$.

We will now show by induction on $d$ that
\begin{equation}\label{eq:ind-claim}
  \begin{bsmallmatrix} n \\ n^{d-2} \end{bsmallmatrix}^{\odot d}\odot \begin{bsmallmatrix} 1 \\ n^{d-1}\end{bsmallmatrix}^{\odot (n-d)}\in \SSS
\end{equation}
for $d=3,4,\dots,n$. We already know that the base case $d=3$ is true.
Suppose for the induction step that 
$\begin{bsmallmatrix} n \\ n^{d-2} \end{bsmallmatrix}^{\odot d}\odot \begin{bsmallmatrix} 1 \\ n^{d-1}\end{bsmallmatrix}^{\odot (n-d)}\in \SSS.$
 Then we have using \eqref{eq:refine} that
  \[
    \begin{bsmallmatrix} n \\ n^{d-1} \end{bsmallmatrix}^{\odot ( d+1)}\odot \begin{bsmallmatrix} 1 \\ n^{d}\end{bsmallmatrix}^{\odot (n-d-1)}=
\begin{bsmallmatrix} n \\ n^{d-2}+\cdots+n^{d-2} \end{bsmallmatrix}^{\odot d}\odot 
\begin{bsmallmatrix} 1+1+\cdots+1 \\ n^{d-1} \end{bsmallmatrix}\odot
\begin{bsmallmatrix} 1 \\ n^{d-1}+\cdots+n^{d-1}\end{bsmallmatrix}^{\odot (n-d)-1}\in \SSS.
\]
This proves \eqref{eq:ind-claim}.

Finally, by setting $d=n$ in \eqref{eq:ind-claim} we obtain the claim $\begin{bsmallmatrix} n \\ n^{n-2}\end{bsmallmatrix}^{\odot n}\in \SSS$.
\end{proof}

\section{Application: Subrank is not additive}\label{sec:not-additive}

We discuss in this section an application of our upper bound \autoref{thm:subrank-upper-3-tensor} on the generic subrank, namely that the subrank is not additive under the direct sum. That is, there are tensors $S,T$ such that $\subrank(T) + \subrank(S) < \subrank(T \oplus S)$. In fact, we obtain a large gap between $\subrank(T) + \subrank(S)$ and $\subrank(T \oplus S)$. The proof relies on the idea of writing the diagonal tensor~$I_n$ as a sum of two generic tensors and on basic properties of the subrank.\footnote{This idea was used earlier in \cite{burg} to show that the ``lower support functionals'' \cite{strassen1991degeneration} are not additive. Their results can be used to find a weaker non-additivity for subrank (with a smaller gap, and only for tensors of order three), namely that there are tensors $S,T \in K^{n,n,n}$ such that $\subrank(T), \subrank(S) \leq n^{2/3 + o(1)}$ and $\subrank(T \oplus S) \geq n$.}

\begin{theorem}\label{thm:non-add}
There are tensors $S,T \in K^{n,n,n}$ such that we have $\subrank(T), \subrank(S) \leq \sqrt{3n-2}$ and $\subrank(T \oplus S) \geq n$.
\end{theorem}

\begin{proof}
By \autoref{thm:subrank-upper-3-tensor} there is a non-empty Zariski-open subset $U \subseteq K^{n,n,n}$ such that for all $T \in U$ we have $\subrank(T) \leq \sqrt{3n-2}$. Recall that $I_n \in K^{n,n,n}$ is the tensor with ones in the diagonal entries and zeroes elsewhere. The subset $I_n - U \subseteq K^{n,n,n}$ is also non-empty and Zariski-open, and thus the intersection $U \cap (I_n - U)$ is non-empty. This means that there is a $T \in U$ such that $I_n - T$ is in $U$. Fix this $T$ and let $S = I_n - T$. Then $\subrank(T), \subrank(S) \leq \sqrt{3n-2}$. For their direct sum, we observe the simple general fact that $T \oplus S \geq T + S$ where the left-hand side is the direct sum and the right-hand side is the coordinate-wise sum. Since subrank is monotone under $\geq$ we find that $\subrank(T \oplus S) \geq \subrank(T + S) = \subrank(I_n) = n$.
\end{proof}

The proof of \autoref{thm:non-add} extends directly so that similarly from \autoref{thm:subrank-upper-k-tensor} we get the following higher-order non-additivity result.

\begin{theorem}
There are $k$-tensors $S,T \in K^{n,\ldots,n}$ such that $\subrank(T), \subrank(S) \leq (kn-(k-1))^{\frac{1}{k-1}}$
while $\subrank(T \oplus S) \geq n$.
\end{theorem}

\section{Open problems}\label{sec:open-problems}

There are several natural open problems that arise from or are closely related to our study and results on the generic subrank in this paper. We briefly list some of these problems here.

\begin{itemize}
    \item While we determine the generic subrank very precisely up to a small additive constant, it is natural to ask whether our upper bound on the generic subrank $\gensubrank(n)$ is exactly tight. We know that this is the case for all small cases ($n \leq 100$). %
    \item Closely related to the above, but for higher-order and unbalanced formats, what is the exact value of the generic subrank $\gensubrank(n_1,\ldots,n_k)$ in $K^{n_1,\ldots, n_k}$ when the dimensions~$n_i$ are not all equal? In particular, is our upper bound tight?
    \item To attack the previous two problems the natural approach would be to use the stronger 
    \autoref{thm:technical-lower-bound-strong}.
    Our current lower bound on the generic subrank uses \autoref{thm:technical-lower-bound} which relies on our constructions of decompositions $\mathcal{X}_1[1] + \mathcal{X}_2[2] + \mathcal{X}_3[3] = K^{r,r,r}$.  \autoref{thm:technical-lower-bound-strong} suggests to instead construct decompositions $\mathcal{X}_1[1] + \mathcal{X}_2[2] + \mathcal{X}_3[3] + W_r = K^{r,r,r}$. What are the ``best'' constructions that can be obtained of this form?
    \item There is a natural approximative version of subrank called \emph{border} subrank\footnote{The difference with subrank is that the restriction preorder is replaced by the \emph{degeneration} preorder, which results in the border subrank $\bsubrank(T)$ of $T$ being the largest number $r$ such that $I_r$ is in the $\GL_n^{\times3}$-orbit closure of $T$.} which plays a central role in algebraic complexity theory (in particular the study of matrix multiplication algorithms).
    The border subrank is at least the subrank. What value %
    does the border subrank take on generic tensors?
    \item What is the value of the generic asymptotic subrank $\asympsubrank(n)$? The state of the art is that $n^{2/3} \leq \asympsubrank(n) \leq n$, so in particular it remains open whether generic tensors have ``full'' generic asymptotic subrank or not. We note that we do know of explicit tensors for which the asymptotic subrank is not full (the so-called W-tensor, for example) and also of explicit tensors for which the asymptotic subrank is full in a non-trivial way (matrix multiplication tensors, for example). The aforementioned lower bound of~$n^{2/3}$ is by Strassen's elegant construction (which makes use of the matrix multiplication tensors). Our results show that one cannot obtain better lower bound on~$\asympsubrank(n)$ by improving the lower bound on the generic subrank $\gensubrank(n)$ since $\gensubrank(n) = \theta(\sqrt{n})$.
    \item With the notation we introduced in \autoref{sec:constructions} on constructions of tensor space decompositions, it is natural to ask what precisely are the elements of $\mathcal{S}$. Many elements can be constructed from simple base cases and the direct sum construction. However, we do not know whether a finite number of generators in this sense suffices to generate all of $\mathcal{S}$.
    \item We found tensors $S,T$ for which there is a large gap between $\subrank(S \oplus T)$ and $\subrank(S) + \subrank(T)$. Is this the largest possible gap? More speculatively, we may ask: is there a general relation between direct sum problems (i.e.~additivity under direct sum) and parameter values at generic instances?
\end{itemize}

\paragraph{Acknowledgements.}
The authors thank Swastik Kopparty for helpful discussions, and JM Landsberg and Jan Draisma for comments.
HD was partially supported by NSF grants IIS-1837985 and DMS-2001460. VM was partially supported by NSF grants CCF-1900460 and the University of Melbourne.
JZ was partially supported by a Simons Junior Fellowship and NWO Veni grant VI.Veni.212.284.

\bibliographystyle{alphaurl}
\bibliography{main}

\end{document}